\documentclass[10pt,journal]{IEEEtran}
\usepackage{amsmath,graphicx}
\usepackage{array}
\usepackage{chngpage}
\usepackage{amsthm}
\usepackage{listings}           
\usepackage{pgfplots}
\usepackage{todonotes}
\usepackage[USenglish]{babel}
\usepackage[sort]{cite}
\usepackage{amsfonts}
\usepackage{color}
\usepackage{subcaption}
\usepackage{fancyhdr} 
\usepackage{amsmath,amsfonts,bm}
\usepackage{mathtools}

\newcommand{\norm}[1]{\left\lVert#1\right\rVert}

\newcommand{\abs}[1]{\left\lvert#1\right\rvert} 
 
\newcommand{\paren}[1]{\left(#1\right)} 
 
\newcommand{\cbracket}[1]{\left\{#1\right\}} 
\renewcommand{\vec}[1]{\bm{#1}}
\newcommand{\mat}[1]{\bm{#1}}
\newcommand{\tran}{^{\mathstrut\scriptscriptstyle\top}} 
\newcommand{\herm}{^{\mathstrut\scriptscriptstyle H}} 
\newcommand{\conj}{*}
\newcommand{\diag}[1]{\mathrm{diag}\left(#1\right)}

\newcommand{\imag}{\imath}

\newcommand{\stateSpace}[1]{\mathfrak{#1}}

\newcommand{\A}[0]{{\mat{\stateSpace{A}}}}

\newcommand{\N}[0]{\stateSpace{{N}}}

\newcommand{\Nset}[0]{\langle N \rangle}

\newcommand{\rank}{\mathrm{rank }}

\newtheorem{theorem}{Theorem}
\newtheorem{lemma}{Lemma}

 \usepackage{mathtools}

 \newcommand\coolleftbrace[2]{#1\left\{\vphantom{\begin{matrix} #2 \end{matrix}}\right.}

\title{On Lossless Feedback Delay Networks}
\author{Sebastian J. Schlecht and Emanu\"{e}l A.P. Habets, ~\IEEEmembership{Senior Member,~IEEE} 
\thanks{S.J. Schlecht and E.A.P. Habets are with the International Audio Laboratories Erlangen (a joint institution of the University of Erlangen-Nuremberg and Fraunhofer IIS), Germany (e-mail: \{sebastian.schlecht,emanuel.habets\}@audiolabs-erlangen.de).}}

\begin{document}

\maketitle
\thispagestyle{fancy}

\begin{abstract}
Lossless Feedback Delay Networks (FDNs) are commonly used as a design prototype for artificial reverberation algorithms. The lossless property is dependent on the feedback matrix, which connects the output of a set of delays to their inputs, and the lengths of the delays. Both, unitary and triangular feedback matrices are known to constitute lossless FDNs, however, the most general class of lossless feedback matrices has not been identified. In this contribution, it is shown that the FDN is lossless for any set of delays, if all irreducible components of the feedback matrix are diagonally similar to a unitary matrix. The necessity of the generalized class of feedback matrices is demonstrated by examples of FDN designs proposed in literature.      
\end{abstract}
\begin{IEEEkeywords}
Feedback Delay Network, Lossless, Diagonal Similarity, Artificial Reverberation.
\end{IEEEkeywords}

\section{Introduction}
\label{sec:Intro}
\IEEEPARstart{F}{eedback} Delay Networks (FDNs) consist of a set of delays and a feedback matrix through which the delay outputs are coupled to the delay inputs (see Fig.~\ref{fig:jot}). FDNs have been proposed by Stautner and Puckette \cite{Stautner1982} and gained popularity as an efficient and flexible way to model and process artificial reverberation \cite{valimaki2012fifty}. Many parametric reverberation algorithms like Schroeder's Cascaded Allpass and Moorer's extension \cite{schroeder1962natural, Moorer:1979hi}, Dattorro's Allpass Feedback Network \cite{Dattorro1997}, Dahl and Jot's Absorbent Allpass FDN \cite{Dahl2000}, can be expressed by an FDN \cite{schlechtParallel,schlecht2014modulation}. Other physical modeling structures like digital waveguide networks (DWN), digital waveguide mesh (DWM), scattering delay networks (SDN) and finite difference time domain (FDTD) simulation are in close relation to the FDN structure \cite{Rocchesso96, karjalainen2004digital, DeSena:2015bb, Murphy:2007hd}.       

It is common practice for FDNs to first design the system to be \emph{lossless}, i.e., all the system poles of the FDN lie on the unit circle \cite{Rocchesso96}. This later ensures a smooth frequency-dependent pole magnitude by simply extending every delay element with an attenuation \cite{Jot1991}. From a practical point of view it is desirable that the delays can be scaled without changing the lossless property as this corresponds to scaling the underlying physical model, e.g. the dimensions of a room model \cite{Rochesso1995} and the corresponding mixing time \cite{Schlecht:2016:mixing}.

From the beginning, unitary and triangular feedback matrices have been found to constitute a lossless FDN independent from the choice of the delays \cite{Stautner1982, Jot1991}. In a more general approach, Gerzon described \emph{unitary networks}, i.e., frequency-dependent unitary matrices of filters, in a feedback loop to constitute a lossless FDN \cite{gerzon1976unitary}. However, a non-unitary, non-triangular feedback matrix can constitute a lossless FDN as well \cite{schlecht2014modulation}. It has been shown that if an FDN is lossless then the feedback matrix has only eigenvalues on the unit circle \cite{Rocchesso96}.

The paper is organized as follows. In Section \ref{sec:Motivation}, we give motivating examples on the relation between feedback matrices and the FDN system poles. Further, we show that the lossless property is dependent on the delays as well. These examples open the question about the definition of the lossless property and its conditions. In Section \ref{sec:LosslessnessCondition}, we derive a new sufficient and necessary condition for FDNs to be lossless independent from the choice of delays. In Section \ref{sec:LosslessTopologies}, some well-known  reverb topologies and their representation as FDNs are given.

\section{Motivation \& Background}
\label{sec:Motivation}
In this section, the FDN is defined and examples on lossless FDNs are given. Then a definition of feedback matrices in lossless FDNs is proposed and the full characterization of lossless FDNs is formulated.

\subsection{Feedback Delay Network}

\begin{figure}[!tb]
\includegraphics[width = 0.46\textwidth]{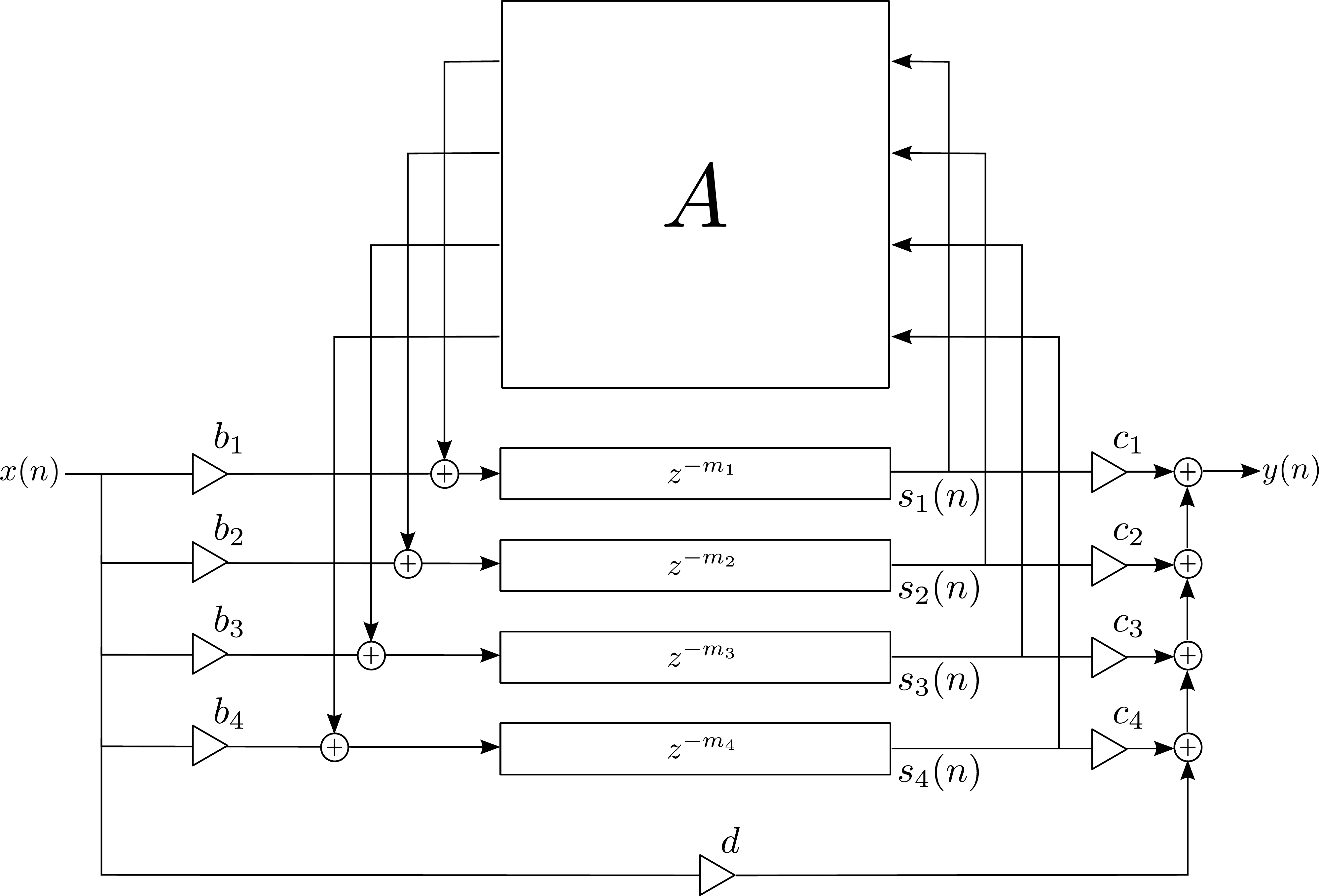}
\caption{Feedback Delay Network (FDN) structure for artificial reverberation as proposed by Stautner and Puckette \cite{Stautner1982} and further developed by Jot and Chaigne \cite{Jot1991}.}
\label{fig:jot}
\end{figure}

The standard FDN is given in time domain by the relation \cite{Rocchesso96}
\begin{subequations}
\begin{align}
	y(n) &=\sum \limits_{i=1}^N c_i \, s_i(n) + d \, x(n) \\
	s_i(n+m_i) &= \sum \limits_{j=1}^N a_{ij} \, s_j(n) + b_i \, x(n),
\end{align}
\label{eq:jotdifference}
\end{subequations}
where $x(n)$ and $y(n)$ are the input and output values respectively and $s_i(n), 1 \leq i \leq N$, are the delay outputs at time sample $n$ \cite{Rocchesso96}. 
The double-indexed $a_{ij}$ form a $N \times N$ \emph{feedback matrix} $\mat{A}$. In a similar manner, we refer to the column vectors of $b_i$'s and $c_i$'s as $\vec{b}$ and $\vec{c}$ respectively. The lengths of the $N$ delays in samples are given by $\vec{m} = [m_1, \dots, m_N] \in \mathbb{N}^N$. The transfer function of the FDN in the z-domain \cite{Rocchesso96} is 

\begin{equation}
H(z) =  \frac{Y(z)}{X(z)} = \vec{c}^T[\mat{D}_{\vec{m}}(z^{-1}) - \mat{A}]^{-1}\vec{b} + d,
\label{eq:fdntrans}
\end{equation}
where $\mat{D}_{\vec{m}}(z) =\diag{[z^{-m_1}, z^{-m_2},\dots,z^{-m_N} ]}$ is the diagonal $N \times N$ \emph{delay matrix}. If the notation is unambiguous, we write $\mat{D}(z)$ instead of $\mat{D}_{\vec{m}}(z)$. The order of the system given in \eqref{eq:fdntrans} is equal to the sum of the delays \cite{Rocchesso96}
\begin{eqnarray}\label{eq:jotorder}
\N &=& \sum_{j=1}^{N} m_j.
\end{eqnarray}
The zeros of the transfer function \eqref{eq:fdntrans} are the roots of the polynomial \cite{Rocchesso96}
\begin{align}
	q_{\mat{A},\vec{b},\vec{c},d,\vec{m}}(z) = \det \left[\mat{A} - \vec{b}\frac{1}{d}\vec{c}\tran - \mat{D}_{\vec{m}}(z^{-1}) \right].	
\label{eq:FDNzerosDefinition}
\end{align}
The poles of the transfer function \eqref{eq:fdntrans} are the roots of the polynomial \cite{Rocchesso96}
\begin{align}
	p_{\mat{A},\vec{m}}(z) &= \det [\mat{D}_{\vec{m}}(z^{-1}) - \mat{A} ].
	\label{eq:charPolyDef}
\end{align}
We call $p_{\mat{A},\vec{m}}(z)$ the generalized characteristic polynomial of $\mat{A}$ with delays $\vec{m}$. Please note that $p_{\mat{A},[1,\dots,1]}(z)$ is the standard characteristic polynomial found in literature \cite{golub2012matrix}. Therefore, the FDN is lossless if $p_{\mat{A},\vec{m}}(z)$ has only unimodular roots, i.e., all roots are of magnitude $1$. We call such a polynomial $p_{\mat{A},\vec{m}}(z)$ to be lossless.

\subsection{Motivating Examples}
  
Let us assume an FDN of dimension $N=2$, with the delays $m_1 = 1, m_2 = 2$ and feedback matrix 
	\begin{eqnarray}
		\mat{A} = \begin{bmatrix}  3 & 2 \\ -4 & -3 \end{bmatrix}. 
	\end{eqnarray}
	The eigenvalue decomposition $\mat{A} = \mat{U}^{-1}\,\mat{\Lambda}\,\mat{U} $ yields
	\begin{eqnarray}
		\mat{U} =  \begin{bmatrix}  2 & 1 \\ 1 & 1 \end{bmatrix}  \textrm{ and }  
		\mat{\Lambda} =  \begin{bmatrix}  1 & 0 \\ 0 & -1 \end{bmatrix},
	\end{eqnarray}
	such that $\mat{U} $ is an invertible matrix of eigenvectors, $\mat{\Lambda}$ is a unimodular diagonal matrix of eigenvalues. The roots of
	\begin{equation}
		\begin{aligned}
		p_{\mat{A},\vec{m}}(z) &= (z^2 - 3)(z^1 + 3) + 8 \\
			&= z^3 - 3z^2 + 3z^1 -1 
		\end{aligned}
	\label{eq:startCounterExamplePositiveCase}
	\end{equation}
	are 
	\begin{align}
		z_{1,2,3} = [1, 1, 1]
	\end{align}
	and therefore the FDN is lossless. However on the contrary,  with $m_2 = 2, m_1 = 1$ the roots of
	\begin{equation}
		\begin{aligned}
			p_{\mat{A},\vec{m}}(z) &= z^3 + 3z^2 - 3z -1  \\
									&= (z^2 + 4 z + z)(z -1)
		\end{aligned}
	\end{equation}
	are
	\begin{align}
		z_{1,2,3} = [1, -2 - \sqrt{3}, -2 + \sqrt{3}]
	\end{align}
	and therefore the FDN is not lossless. Consequently, $\mat{A}$ having unimodular eigenvalues and linearly independent eigenvectors is insufficient for the FDN to be lossless. This example contradicts the claim presented in \cite[Eq.~(29)]{Rocchesso96} that $\mat{A}$ having unimodular eigenvalues and linearly independent eigenvectors is sufficient for $p_{\mat{A},\vec{m}}(z)$ to be lossless. In Appendix \ref{sec:FailedProofRocchesso}, the technical details of the flawed proof in \cite{Rocchesso96} are given.

As an additional note, the stability of an FDN is not even ensured if all eigenvalues of $\mat{A}$ lie within the unit circle. For  
\begin{eqnarray}
\mat{A} = \mat{U}^{-1}\,\frac{\mat{\Lambda}}{2}\,\mat{U} = \begin{bmatrix}  1.5 & 1 \\ -2 & -1.5 \end{bmatrix}
\end{eqnarray}
we have
\begin{align}
p_{\mat{A},\vec{m}}(z) &= z^3 + 1.5z^2 - 1.5z - 0.25 .
\end{align}
The largest eigenvalue magnitude is then approximately $2.145$ and therefore the FDN is unstable.

These examples illustrate that
\begin{itemize}
\item there are non-unitary, non-triangular $\mat{A}$, for which the FDN is lossless.
\item there are $\mat{A}$ with only unimodular eigenvalues, for which the FDN is not lossless.
\item the losslessness of the FDN can be dependent on $\vec{m}$.
\item the eigenvalues of $\mat{A}$ being within the unit circle is not a sufficient condition for stability.
\end{itemize}

In the following, the existing results on lossless FDNs are summarized and a precise definition of the main result is stated.
 
\subsection{Lossless FDN}
The FDN's losslessness should not depend on the choice of the delays $\vec{m}$, but only on the feedback matrix $\mat{A}$. In other words, we want to characterize all matrices $\mat{A}$ such that $p_{\mat{A},\vec{m}}(z)$ is lossless for any choice of $\vec{m} \in \mathbb{N}^N$. We call such a matrix to be \emph{unilossless}.\footnote{Please note that matrices $\mat{A}$ with unimodular eigenvalues and linearly independent eigenvectors have been named lossless in \cite{Rocchesso96}. In this manuscript, this terminology is avoided to reduce confusion and to emphasize the dependency on the choice of the delays.}

\begin{figure}[!tb]
	\centering 
	\includegraphics[width=0.4\textwidth]{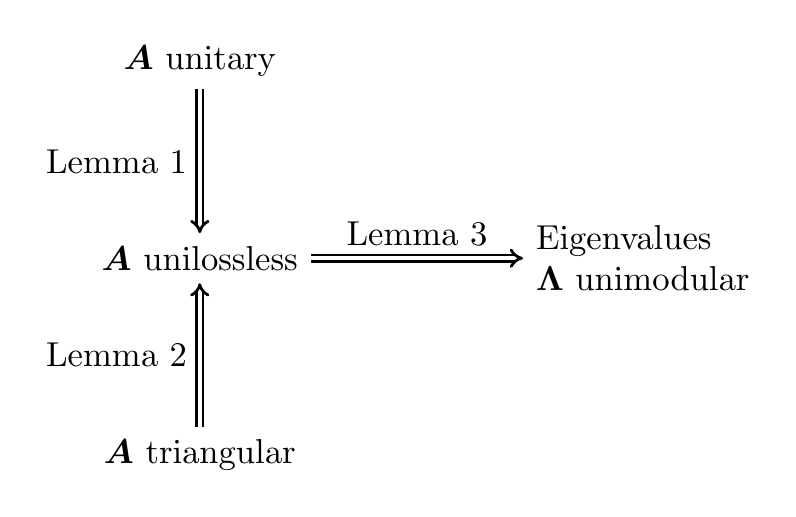}
	\caption{Implication graph for unilossless matrices $\mat{A}$ in the literature.} 
	\label{fig:priorArtDependency}
\end{figure}

It has be shown in \cite{Stautner1982} that $\mat{A}$ being \emph{unitary}, i.e., $\mat{A}\mat{A}\herm = \mat{I}$, is a sufficient condition for $\mat{A}$ to be unilossless. For completeness, we give an alternative proof.

\begin{lemma}
	Any unitary matrix $\mat{A}$ is unilossless.
	\label{th:AnyUnitaryIsUnilossless}
\end{lemma} 

\begin{proof}
	Firstly, $|p_{\mat{A},\vec{m}}(0)|=|\det(−\mat{A})|=1$, hence there is no root of $p_{\mat{A},\vec{m}}$ which is zero.

	Secondly, if $z \neq 0$, such that $p_{\mat{A},\vec{m}}(z)=0$ is equivalent to the existence of $\vec{v}\neq \vec{0}$ such that $\mat{A}\vec{v}=\mat{D}(z^{-1})\vec{v}$. Writing $z^{-1}=\rho e^{\imag t}$, where $\imag = \sqrt{-1}$, we obtain $\mat{A}\vec{v}=\mat{D}(\rho)\mat{D}(e^{\imag t})\vec{v}$ . Using the fact that $\mat{D}(e^{\imag t})$ is unitary, we obtain
		\begin{align}
			\norm{\vec{v}}_2=\norm{\mat{D}(e^{-\imag t})\mat{A}\vec{v}}_2=\norm{\mat{D}(\rho)\vec{v}}_2.
			\label{eq:proofUnitaryLosslessEquivalence}
		\end{align}
	As all delays $\vec{m} \geq 1$, two cases can be identified for all $\vec{v}$ and $\rho>0$:
	\begin{equation}
	\begin{aligned}
		\norm{\mat{D}(\rho)\vec{v}}_2 \geq \rho \norm{\vec{v}}_2 & \quad \textrm{ if } \rho<1 \\
		\norm{\mat{D}(\rho)\vec{v}}_2 \leq \rho \norm{\vec{v}}_2 & \quad \textrm{ if } \rho>1.
	\end{aligned}
	\end{equation}
	the result follows from \eqref{eq:proofUnitaryLosslessEquivalence}.
\end{proof}

A second class of unilossless matrices has been identified by Jot in \cite{Jot1991}, the triangular matrices $\mat{A}_T$ with unimodular main diagonal, i.e., $|t_{ii}| = 1$ for all $i$ and $t_{ij}=0$ either for $i>j$ or $i<j$.

\begin{lemma}
	Any triangular matrix with unimodular main diagonal is unilossless.
	\label{th:unilosslessTriangular}
\end{lemma}  

\begin{proof}
	Given triangular matrix $\mat{A}_T$ with $|t_{ii}| = 1$, then $p_{\mat{A}_T,\vec{m}}(z) = \prod_{i=1}^N (z^{m_i} - t_{ii})$. The roots of $(z^{m_i} - t_{ii})$ are unimodular and therefore all roots of $p_{\mat{A}_T,\vec{m}}(z)$ are unimodular and $\mat{A}_T$ is unilossless.
\end{proof}

A necessary condition on unilossless matrices has been given in \cite{Rocchesso96}.
\begin{lemma}
	Any unilossless matrix $\mat{A}$ has only unimodular eigenvalues.
	\label{th:losslessUnimodular}
\end{lemma}

\begin{proof}
	As $\mat{A}$ is unilossless the roots of $p_{\mat{A},\vec{m}}(z)$ are unimodular for any $\vec{m}$. Therefore, the roots of $p_{\mat{A},[1,\dots,1]}(z)$, which are the eigenvalues of $\mat{A}$, are unimodular.
\end{proof}

The converse of Lemma \ref{th:losslessUnimodular} can be stated with an additional condition that all delays are equal. 

\begin{lemma}
	Let $\mat{A}$ have only unimodular eigenvalues and all delays are equal, i.e., $m_1 = \dots = m_N$, then $p_{\mat{A},\vec{m}}(z)$ is lossless.
	\label{th:losslessEqualDelay}
\end{lemma}

\begin{proof}
	The roots of the characteristic polynomial $\zeta \mat{I} - \mat{A}$ with $\zeta \in \mathbb{C}$ are the eigenvalues of $\mat{A}$ and therefore unimodular. The substitution $z^m_1 = \dots = z^m_N = \zeta$ maps the unit circle onto itself. Consequently, the roots of $p_{\mat{A},\vec{m}}(z)$ are unimodular.	
\end{proof}

For natural reverberation, the echo density usually increases exponentially with time. When the delay lengths are equal, i.e., $m_1 = m_2 = ... = m_N$, there is no increase of the echo density across time such that the reverberation will be unnatural \cite[pp.~109]{Kuttruff2000}.  The delays $\vec{m}$ employed in Lemma~\ref{th:losslessEqualDelay} are therefore a special case. Figure~\ref{fig:priorArtDependency} gives an overview of the known results on unilossless matrices as described in the Lemmas \ref{th:AnyUnitaryIsUnilossless}, \ref{th:unilosslessTriangular} and \ref{th:losslessUnimodular}. The next section is devoted to the derivation of the necessary and sufficient conditions of unilossless matrices.

\section{Unilossless Matrices}
\label{sec:LosslessnessCondition}

Before the main result of the present work on the characterization of unilossless matrices is given, we introduce the notion of independent matrix components. A matrix $\mat{A}$ is \emph{reducible} if there is a permutation matrix $\mat{P}$ with
\begin{equation}
\begin{aligned}
	\mat{P}^{-1} \mat{A} \mat{P} = 
	\begin{bmatrix}
		\mat{E} & \mat{G} \\
		\mat{0}	& \mat{F}
	\end{bmatrix},
	\label{eq:permutationCojugate}
\end{aligned}
\end{equation}
with $\mat{E}$ and $\mat{F}$ square. In other words, a matrix is reducible if it is permutation conjugate to a block triangular matrix.

\begin{theorem}
	A matrix $\mat{A}$ is unilossless if and only if either
	\begin{enumerate}
		\item $\mat{A}$ is reducible and all irreducible components are unilossless or
		\item there exists a non-singular diagonal matrix $\mat{E}$ with $\mat{A} \mat{E} \mat{A}\herm = \mat{E}$.
	\end{enumerate} 
	\label{th:sufficientNecessaryUnilossless}
\end{theorem}

In this section, we explore the sufficient and necessary conditions for the feedback matrix to be unilossless, and in particular give a proof for Theorem \ref{th:sufficientNecessaryUnilossless}. Figure~\ref{fig:theoremDependency} sketches the dependencies of the following theorems and may guide the reader towards the main result.

\subsection{Generalized Characteristic Polynomial}
To derive the conditions on the generalized characteristic polynomial to have only unimodular roots, it is helpful to give a closed form of $p_{\mat{A},\vec{m}}(z)$. This closed form of the generalized characteristic polynomial $p_{\mat{A},\vec{m}}(z)$ in \eqref{eq:charPolyDef} was proven by the authors in \cite{schlecht2014modulation} in terms of principal minors. A principal minor $\det \mat{A}(I)$ of a matrix $\mat{A}$ is the determinant of a submatrix $\mat{A}(I)$ with equal row and column indices $I \subset \Nset$.  $\abs{I}$ indicates the cardinality of set $I$. The set of all indices is denoted by $\Nset = \{1,2,\dots,N\}$ and $I^c$ is the relative complement in $\Nset$, i.e., $I^c = \Nset \setminus I$. 

\begin{theorem}
For a given feedback matrix $\mat{A}$ and delays $\vec{m}$, the generalized characteristic polynomial $p_{\mat{A},\vec{m}}$ is given by
\begin{align}
	p_{\mat{A},\vec{m}}(z) &= \sum_{k = 0}^{\N} c_k \, z^k   \\
	c_k &= 	
			\begin{cases} 
				\sum_{I \in I_k} (-1)^{N - |I|}\det \bold{A}(I^c),  & \text{for } I_k \neq \emptyset\\ 
  				0, & \text{otherwise }
			\end{cases}
\label{eq:cPcoefficients}			
\end{align}
where $I_k = \{I \subset \Nset | \sum_{i \in I} m_i\ = k\}$ is the set of index sets $I$ such that the delays with indices in $I$ sum up to $k$. The system order $\N$ is defined in \eqref{eq:jotorder}. 
\label{th:charPoly}
\end{theorem}

There are $2^N$ different principal minors, and for many choices of $\vec{m}$ the coefficients $c_k$ are sums of multiple minors. Still, the principal minors can be assigned uniquely to the coefficients, for example with exponential delays $\vec{m} = [1, 2, 4, \dots, 2^{N-1}]$.

However, for real matrices $\mat{A}$ there are no more than $N^2$ degrees of freedom defined by the square matrix entries, which implies that there is a strong dependency between the principal minors for larger matrices. In fact, it can be shown that there are only $N^2 - N + 1$ degrees of freedom \cite{stouffer1924independence} because of a similarity invariance which we discuss in following Section~\ref{sec:ExtendingTheSetOfUnilossless}.

\begin{figure}[!tb]
	\centering 
	\includegraphics[width=0.5\textwidth]{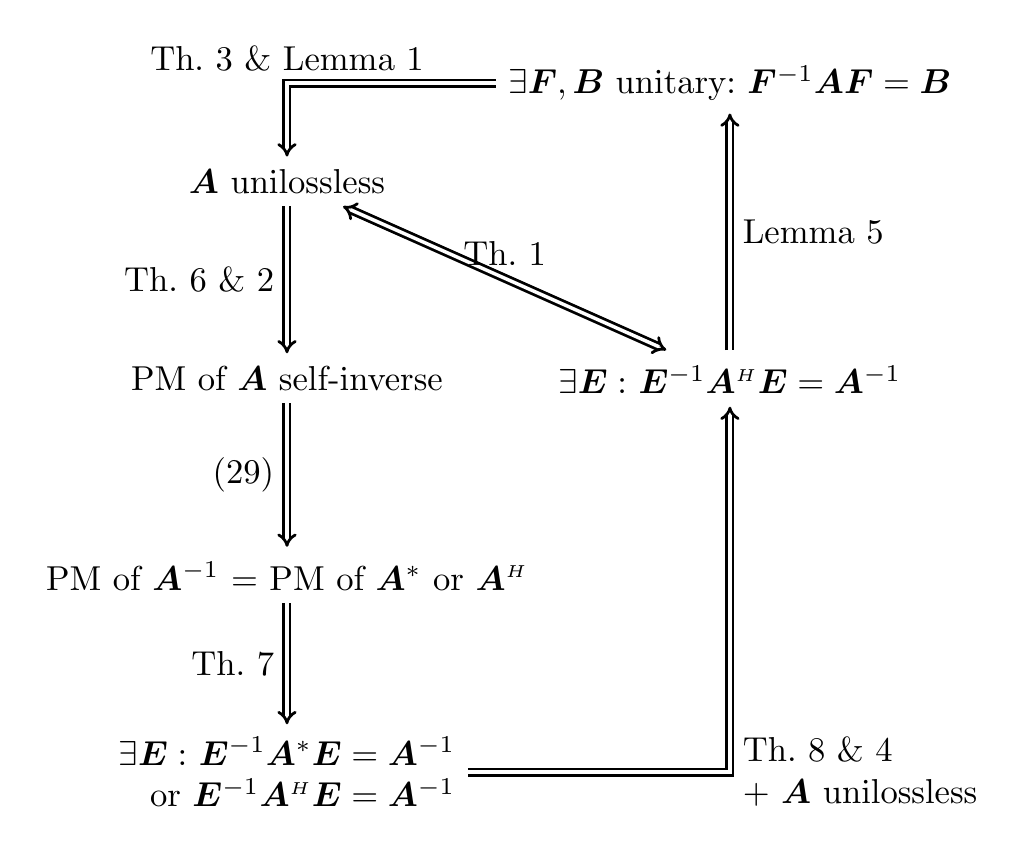}
	\caption{Implication graph of theorems for unilossless and irreducible $\mat{A}$. All $\mat{E}$ and $\mat{F}$ are non-singular diagonal matrices.} 
	\label{fig:theoremDependency}
\end{figure} 

\subsection{Extending the Set of Unilossless Matrices}
\label{sec:ExtendingTheSetOfUnilossless}

We will discuss three techniques to extend the known class of unilossless matrices, which are described by the Lemmas \ref{th:AnyUnitaryIsUnilossless}, \ref{th:unilosslessTriangular} and \ref{th:losslessUnimodular}.  

\subsubsection{Diagonal Similarity} We show that conjugating a unilossless matrix with a diagonal matrix does not alter the principal minors and therefore also not the unilossless property.
\begin{theorem}
	Given a non-singular diagonal matrix $\mat{E}$ and any matrix $\mat{A}$. The matrices $\mat{A}$ and $\mat{E}^{-1} \mat{A} \mat{E}$ have the same principal minors.
	\label{th:conjugateDiagonalPMInvariance}
\end{theorem}

\begin{proof}
	We show that the generalized characteristic polynomials of $\mat{A}$ and $\mat{E}^{-1} \mat{A} \mat{E}$ are identical. 
	\begin{equation}
\begin{aligned}
	p_{\mat{E}^{-1} \mat{A} \mat{E},\vec{m}}(z) 
	&= \det \left[\mat{D}(z^{-1}) - \mat{E}^{-1} \mat{A} \mat{E} \right] \\
	&= \det \left[\mat{E}^{-1} \mat{D}(z^{-1}) \mat{E} - \mat{E}^{-1} \mat{A} \mat{E}  \right] \\
	&= \det \mat{E}^{-1} \det \left[ \mat{D}(z^{-1}) - \mat{A}  \right] \det \mat{E} \\
	&= p_{\mat{A},\vec{m}}(z) 
\end{aligned}
	\end{equation}
\end{proof}

Consequently, all diagonal conjugations of a unilossless matrix are unilossless again. As we have noted earlier, the degree of freedom of the principal minors for real matrices is $N^2$ minus the invariance of the diagonal conjugation of $N-1$ degrees. 

It is interesting to note that the diagonal similarity is rather a property of the system zeros than of the system poles. As noted in Theorem \ref{th:conjugateDiagonalPMInvariance}, the principal minors $\mat{A}$ and $\mat{E}^{-1} \mat{A} \mat{E}$ are equal and therefore also the system poles are invariant under diagonal similarity. From \eqref{eq:FDNzerosDefinition}, we have
\begin{equation}
\begin{aligned}
	&q_{\mat{E}^{-1} \mat{A} \mat{E},\vec{b},\vec{c},d,\vec{m}}(z) \\
	&= \det \left[\mat{E}^{-1} \mat{A} \mat{E} - \vec{b}\frac{1}{d}\vec{c}\tran - \mat{D}(z^{-1}) \right] \\
	&= \det \mat{E}^{-1} \det \left[\mat{A} - \mat{E} \vec{b}\frac{1}{d}\vec{c}\tran \mat{E}^{-1} - \mat{D}(z^{-1}) \right] \det \mat{E} \\
	&= \det \left[\mat{A} - \vec{b}'\frac{1}{d} {\vec{c'}}\tran- \mat{D}(z^{-1}) \right] 
	= q_{\mat{A},\vec{b}',\vec{c}',d,\vec{m}}(z)
\end{aligned}
\end{equation}
with $\vec{b}' = \mat{E} \vec{b}$ and $\vec{c}' = (\mat{E}^{-1})\tran \vec{c}$. Therefore, the diagonal conjugation effects only the system feed and observation gains $\vec{b}$ and $\vec{c}$, respectively.

\subsubsection{Rotation along the Unit Circle}
\label{sec:dummy}
A further invariance can be derived from the fact that the generalized characteristic polynomial can be rotated along the unit circle without changing the root magnitudes. Such a rotation can be realized by the following theorem.

\begin{theorem}
	Given any feedback matrix $\mat{A}$ and any delay matrix $\mat{D}(z)$. The generalized characteristic polynomial $p_{\mat{D}(\gamma)\mat{A},\vec{m}}(z)$ is lossless for any $\gamma \in \mathbb{C}$ with $|\gamma| = 1$ if and only if $p_{\mat{A},\vec{m}}(z)$ is lossless.
	\label{th:RotationAlongUnitCircle}
\end{theorem} 

\begin{proof}
	Substitute $z = \zeta \gamma$, where $\zeta \in \mathbb{C}$. Then
	\begin{equation}
	\begin{aligned}
		p_{\mat{A},\vec{m}}(z) &= p_{\mat{A},\vec{m}}(\zeta\gamma) \\
			&= \det [\mat{D}(\zeta^{-1}\gamma^{-1}) - \mat{A} ] \\
			&= \det [\mat{D}(\gamma^{-1}) \mat{D}(\zeta^{-1})  - \mat{A} ] \\
			&= \det [\mat{D}(\zeta^{-1}) - \mat{D}(\gamma) \mat{A} ] \\
			&= p_{\mat{D}(\gamma)\mat{A},\vec{m}}(\zeta)
	\end{aligned}	
	\end{equation}
	Hence, $r \in \mathbb{C}$ is a root of $p_{\mat{D}(\gamma)\mat{A},\vec{m}}(z)$ if and only if $r \gamma$ is a root of $p_{\mat{A},\vec{m}}(z)$.	
\end{proof}

This invariance does not translate to unilossless matrices $\mat{A}$ as $\mat{D}(\gamma)$ is dependent on the choice of the delays $\vec{m}$. Theorem 4 nevertheless plays an important role later in the proof the main Theorem~\ref{th:sufficientNecessaryUnilossless}.

\subsubsection{Reducible Matrices}
We show here that the losslessness of systems which can be decomposed into independent components can be determined from its parts. We show that the lossless property of reducible matrices can be deduced from the lossless property of the diagonal blocks.

\begin{theorem}
	Given a reducible matrix $\mat{A}$ and the permutation conjugate like in \eqref{eq:permutationCojugate}. Then there is a partition $\vec{m}_1$, $\vec{m}_2$ of $\vec{m}$ with 
	\begin{align}
		p_{\mat{A},\vec{m}}(z) = p_{\mat{E},\vec{m}_1}(z) p_{\mat{F},\vec{m}_2}(z).
	\end{align}
	\label{th:ReducibleUnilossless}
\end{theorem}

\begin{proof}
We show that the conjugate permutation of the matrix is equivalent to a permutation of the delays:
	\begin{equation}
\begin{aligned}
	p_{\mat{P}^{-1} \mat{A} \mat{P},\vec{m}}(z) 
		&= \det \left[\mat{D}_{\vec{m}}(z^{-1}) - \mat{P}^{-1} \mat{A} \mat{P} \right] \\
		&= \det \left[\mat{P}^{-1}\mat{D}_{\vec{m}\mat{P}^{-1}}(z^{-1}) \mat{P} - \mat{P}^{-1} \mat{A} \mat{P}  \right] \label{eq:permutationCommute} \\
		&= p_{\mat{A},\vec{m} \mat{P}^{-1} }(z),
\end{aligned}
	\end{equation}
where the commutation in \eqref{eq:permutationCommute} follows from Lemma 2.2 in \cite{Stuart:1994gt}. Further, we show that the generalized characteristic polynomial can be decomposed as	
	\begin{equation}
\begin{aligned}	
	p_{\mat{P}^{-1} \mat{A} \mat{P},\vec{m}\mat{P}}(z) 
		&= \det \left[\mat{D}_{\vec{m}\mat{P}}(z^{-1}) - \mat{P}^{-1} \mat{A} \mat{P} \right] \\
		&= \det \left[\mat{D}_{\vec{m}\mat{P}}(z^{-1}) -  \begin{bmatrix}
		\mat{E} & \mat{G} \\
		\mat{0}	& \mat{F}
	\end{bmatrix} \right] \\
	 	&= p_{\mat{E},\vec{m}_1}(z) p_{\mat{F},\vec{m}_2}(z),
\end{aligned}
	\end{equation}
	where $[\vec{m}_1, \vec{m}_2] = \vec{m}\mat{P}$ and follows from the determinant formula for block matrices \cite{brualdi1983determinantal}. Together the proof concludes with
	\begin{align}
		p_{\mat{A},\vec{m}}(z) = p_{\mat{P}^{-1} \mat{A} \mat{P},\vec{m}\mat{P}}(z) = p_{\mat{E},\vec{m}_1}(z) p_{\mat{F},\vec{m}_2}(z).
	\end{align}
	
	 \end{proof}

Consequently, the matrix $\mat{A}$ is unilossless if and only if $\mat{E}$ and $\mat{F}$ are unilossless. This also gives a more general argument for the unilossless triangular matrices characterized in Lemma \ref{th:unilosslessTriangular}. For the intuition, we remark that the diagonal blocks correspond to strongly connected components of the signal graph, i.e., a subset of delays where every delay is connected to each other by a chain of non-zero gains in the feedback matrix \cite{Cormen:2001uw}. The off-diagonal blocks correspond to connections between the strongly connected components. Because of the block triangular form, once a signal leaves a strongly connected component it never returns, which explains the independent stability property of these components.

\subsection{Characterizing the Set of Unilossless Matrices}
In this section, we develop the necessary conditions for unilossless matrices, which eventually results in a full characterization.

\subsubsection{Roots on the Unit Circle}
The starting point is a classical theorem on polynomials which have all roots on the unit circle. In the following, $p'$ denotes the derivative of the polynomial $p$.

\begin{theorem}[A. Cohn, \cite{Cohn:1922fk}]
All zeros of a polynomial $p$ lie on the unit circle if and only if
\begin{enumerate}
\item $p$ is self-inversive 
\item all zeros of $p'$ lie in or on the unit circle.
\end{enumerate}	
\label{th:reciprocalCohn}
\end{theorem}
A polynomial $p(z) = \sum_{k=0}^{\N} c_k z^k$ of order $\N$ is \emph{self-inversive} if
\begin{align}
\exists \epsilon \in \mathbb{C}, \forall j : c_{\N-j} = \epsilon \, c_{j}^\conj.
\end{align}

Criterion 2 of Theorem \ref{th:reciprocalCohn} is similarly difficult to verify as the original problem, however Criterion 1 gives a simple starting point to create a necessary condition on unilossless matrices. The main advantage of Criterion 1 is that it is independent from the choice of the delays.

\subsubsection{Self-Inversive Principal Minors}
We derive the conditions on $\mat{A}$ such that $p_{\mat{A},\vec{m}}$ is self-inversive.  For this, we recall Jacoby's identity for invertible matrices $\mat{A}$ \cite{brualdi1983determinantal}: Given index set $I \subset \Nset$
\begin{align}
	\det \mat{A}^{-1}(I) = \frac{\det \mat{A}(I^c)}{\det \mat{A}}.
	\label{eq:jacoby}
\end{align}
For $\mat{A}$ to be unilossless, $p_{\mat{A},\vec{m}}$ has to be self-inversive for all $\vec{m}$, hence also for $\vec{m} = [1, 2, \dots, 2^{N-1}]$ implying that $I_k$ is unique for every $k$ in \eqref{eq:cPcoefficients}. Consequently as $c_{\N} = \det \mat{A}(\emptyset) = 1$,
\begin{equation}
\begin{aligned}
	\det \mat{A} &= (-1)^{N} c_{0} = (-1)^{N} \epsilon \, c_{\N}^\conj \\
	&= \epsilon \, (-1)^{N} \det \mat{A}(\emptyset) = (-1)^{N} \epsilon. 
\end{aligned}	
\end{equation}	
and because of $I_k = I^c_{\N-k}$ and $|I^c_k|+|I_k| = N$
\begin{equation}
\begin{aligned}
	\det \mat{A}(I_k) &= \det \mat{A}(I^c_{\N-k}) = (-1)^{N-|I_{\N - k}|} c_{\N-k} = \\
	&= (-1)^{N-|I^c_k|} \epsilon \, c_{k}^\conj \\
	&= (-1)^{2N-|I^c_k|-|I_k|} \epsilon \, \det \mat{A}^\conj (I^c_{k})  \\
	&= (-1)^{3N-|I^c_k|-|I_k|} \epsilon \, \det \mat{A}^\conj \, \det \mat{A}^{-\conj}(I_k) \\
	&= (-1)^{2N} \det \mat{A}^{-\conj}(I_k) \\
	&= \det \mat{A}^{-\conj}(I_k)
	\label{eq:reciprocalPrincipalMinors}
\end{aligned}	
\end{equation}	

In other words, $p_{\mat{A},\vec{m}}$ is self-inversive for all $\vec{m}$ only if the principal minors of $\mat{A}$ are the complex conjugate principal minors of $\mat{A}^{-1}$. This holds for $\mat{A}$ with 
\begin{align}
\bold{A}^{-1} = \bold{A}^\conj \textrm{ or } \bold{A}^{-1} = \bold{A}\herm.
\label{eq:inverseMinors}
\end{align}
A matrix $\bold{A}^{-1} = \bold{A}^\conj$ is called \emph{conjugate-involutory}. In the following, we show that these two options are essentially exhaustive if we additionally allow diagonal similarity.  

\subsubsection{Diagonal Similarity and Principal Minors}
In the following, we characterize the matrices with property \eqref{eq:reciprocalPrincipalMinors}. The following theorem is build upon results by Hartfiel and Loewy \cite{Hartfiel1984, loewy1986principal}.
\begin{theorem}[]
	Let $\mat{A}$ be irreducible and invertible.	 If $\mat{A}^\conj$ and $\mat{A}^{-1}$ have equal principal minors then there exists a non-singular diagonal matrix $\mat{E}$ with either $\mat{E}^{-1} \mat{A}^\conj \mat{E} = \mat{A}^{-1}$ or $\mat{E}^{-1} \mat{A}\herm \mat{E} = \mat{A}^{-1}$.
	 \label{th:diagonalSimilarInverseGeneral}
\end{theorem}

\begin{proof}
The proof can be found in Appendix \ref{sec:ProofOfHartfielLoewyInverse}.
\end{proof}

Theorem \ref{th:diagonalSimilarInverseGeneral} gives a characterization of \eqref{eq:inverseMinors} with additional diagonal similarity. This result can be expressed also in terms of diagonal similarity of the matrices themselves.

\begin{lemma}
	If $\mat{E}^{-1} \mat{A}^\conj \mat{E} = \mat{A}^{-1}$ then there exists a diagonal non-singular matrix $\mat{F}$ with $\mat{B} = \mat{F}^{-1} \mat{A} \mat{F}$ such that $\mat{B}^{-1} = \mat{B}^\conj$. Also, if $\mat{E}^{-1} \mat{A}\herm \mat{E} = \mat{A}^{-1}$ then there exists a diagonal non-singular matrix $\mat{F}$ with $\mat{B} = \mat{F}^{-1} \mat{A} \mat{F}$ such that $\mat{B}^{-1} = \mat{B}\herm$.
	\label{th:diagSimilarityAnalog}  
\end{lemma}

\begin{proof}
Let $\mat{B} = \mat{F}^{-1} \mat{A} \mat{F}$ with $\mat{F}^\conj \mat{F}^{-1} = \mat{E}$. Given $\mat{E}^{-1} \mat{A}^\conj \mat{E} = \mat{A}^{-1}$, then equivalently $\mat{A}^\conj \mat{E} \mat{A} = \mat{E}$ and 	
\begin{equation}
\begin{aligned}
		\mat{B}^\conj \, \mat{B} &= \mat{F}^{-\conj} \mat{A}^\conj \mat{F}^\conj \mat{F}^{-1} \mat{A} \mat{F} 
		= \mat{F}^{-\conj} \mat{A}^\conj \mat{E} \mat{A} \mat{F} \\
		&= \mat{F}^{-\conj} \mat{E} \mat{F} = \mat{I}.
\end{aligned}
	\end{equation}

Let $\mat{B} = \mat{F}^{-1} \mat{A} \mat{F}$ with $\mat{F}^{-\conj} \mat{F}^{-1} = \mat{E}$. Given $\mat{E}^{-1} \mat{A}\herm \mat{E} = \mat{A}^{-1}$, then equivalently $\mat{A}\herm \mat{E} \mat{A} = \mat{E}$ and 	
\begin{equation}
\begin{aligned}
		\mat{B}\herm \, \mat{B} &= \mat{F}^{\conj} \mat{A}\herm \mat{F}^{-\conj} \mat{F}^{-1} \mat{A} \mat{F} 
		= \mat{F}^{\conj} \mat{A}\herm \mat{E} \mat{A} \mat{F} \\
		&= \mat{F}^{\conj} \mat{E} \mat{F} = \mat{I}.
\end{aligned}
\end{equation}
	
\end{proof}

Because of Theorem~\ref{th:diagonalSimilarInverseGeneral} and Lemma~\ref{th:diagSimilarityAnalog}, any irreducible and invertible matrices having self-inversive principal minors are either diagonally similar unitary or conjugate-involutory matrices. More precisely, if $\mat{B}$ is a unilossless and irreducible matrix, then there is a non-singular matrix $\mat{E}$ with $\mat{A} = \mat{E}^{-1} \mat{B} \mat{E}$ such that $\mat{A}$ satisfies \eqref{eq:inverseMinors}.

\subsubsection{Irreducible Unilossless is Necessarily Diagonally Similar Unitary}
Theorem \ref{th:diagonalSimilarInverseGeneral} shows that any irreducible unilossless matrix is either diagonally similar to a unitary or conjugate-involutory matrix. The following theorem shows that conjugate-involutory, but non-unitary matrices cannot be unilossless.

\begin{theorem}
	Let $\mat{A}$ be an irreducible unilossless matrix then there exists a non-singular diagonal matrix $\mat{E}$ with $\mat{A} \mat{E} \mat{A}\herm = \mat{E}$.
	\label{th:unilosslessisUnitary}
\end{theorem}

\begin{proof}
	Let  $\mat{A}$ be a irreducible and unilossless matrix. From Theorem \ref{th:diagonalSimilarInverseGeneral} follows that there exists a non-singular diagonal matrix $\mat{E}$ with $\mat{A} \mat{E} \mat{A}^\conj = \mat{E}$ or $\mat{A} \mat{E} \mat{A}\herm = \mat{E}$.
	\begin{enumerate}
		\item Assume $\mat{A} \mat{E} \mat{A}\herm = \mat{E}$: \checkmark
		\item Assume $\mat{A} \mat{E} \mat{A}^\conj = \mat{E}$: 
	 	
	 	Given a unimodular $\gamma$ and $\vec{m}$ with $m_i \neq m_j$ for all $i \neq j$. With Theorem \ref{th:RotationAlongUnitCircle} it follows $p_{\mat{A},\vec{m}}(z)$ is lossless if and only if $p_{\mat{A}\mat{D}(\gamma),\vec{m}}(z)$ is lossless. We write $\mat{B} = \mat{A}\mat{D}(\gamma)$. If $p_{\mat{B},\vec{m}}(z)$ is lossless, then $\mat{B} \mat{E} \mat{B}^\conj = \mat{E}$ or $\mat{B} \mat{E} \mat{B}\herm = \mat{E}$.
	 	\begin{enumerate}
	 		\item Assume $\mat{B} \mat{E} \mat{B}\herm = \mat{E}$:
	 			\begin{align}
	 			\mat{E} = \mat{B} \mat{E} \mat{B}\herm = \mat{A} \mat{D}(\gamma) \mat{E} \mat{D}\herm(\gamma) \mat{A}\herm = \mat{A}  \mat{E}  \mat{A}\herm
	 			\end{align}
	 		\item Assume $\mat{B} \mat{E} \mat{B}^\conj = \mat{E}$:
	 	\begin{equation}
	 	\begin{aligned}
	 		\mat{B} \mat{E} \mat{B}^\conj &= \mat{E} \\
	 		\mat{E}^{-1} \mat{B} \mat{E} &= \mat{B}^{-\conj} \\
	 		(\mat{D}(\gamma)\mat{E})^{-1} \mat{A} \mat{D}(\gamma) \mat{E} &= \mat{A}^{-\conj}
		 \end{aligned}
		 \label{eq:rotatedLosslessConjugateInvolutory}
	 	\end{equation}
	 	
	 	We write $\mat{A}^{-\conj} = [a'_{ij}]_{N\times N}$. Because of Assumption 2, 
	 	\begin{align}
	 		\mat{E}^{-1} \mat{A} \mat{E} &= \mat{A}^{-\conj}
	 	\end{align}
	 	we write
	 	\begin{align}
	 		\frac{e_i}{e_j} a_{ij} = a'_{ij}.
	 	\end{align}
	 		 	
	 	From \eqref{eq:rotatedLosslessConjugateInvolutory}, we have
	 	\begin{align}
	 		\frac{\gamma^{m_i} e_i}{\gamma^{m_j} e_j} a_{ij} = a'_{ij}.
	 	\end{align}
	 	And together, this gives
	 	\begin{align}
	 		\left(\frac{\gamma^{m_i}}{\gamma^{m_j}}-1\right) \frac{e_i}{e_j} a_{ij} = 0.
	 	\end{align}
	 	As $\gamma^{m_i} \neq \gamma^{m_j}$ and $e_i \neq 0$, we have $a_{ij} = 0$ for all $i \neq j$. Then $\mat{A}$ is a diagonal matrix, such that $\mat{A}^\conj = \mat{A}\herm$, and therefore $\mat{A} \mat{E} \mat{A}\herm = \mat{E}$.
	 	\end{enumerate}
	 	\end{enumerate}
\end{proof}
Given any matrix $\mat{B}$, it is possible to test algorithmically whether $\mat{B}$ is diagonally similar to a unitary matrix and to determine the diagonal similarity matrix. More details can be found in \cite{Engel:2006bl}. 

\subsection{Proof of the Main Result}
Figure~\ref{fig:theoremDependency} sketches the dependencies of the theorems we have proven so far for the irreducible matrices $\mat{A}$ and the relation to the main result Theorem \ref{th:sufficientNecessaryUnilossless}. We conclude this section with the proof of Theorem \ref{th:sufficientNecessaryUnilossless}.

\begin{proof}[Proof of Theorem \ref{th:sufficientNecessaryUnilossless}]
	Both cases can be treated independently.
	\begin{enumerate}
		\item If $\mat{A}$ is reducible, the result follows from Theorem \ref{th:ReducibleUnilossless}.
		\item Let $\mat{A}$ be irreducible. If $\mat{A}$ is unilossless then there is a non-singular diagonal matrix $\mat{E}$ with $\mat{A} \mat{E} \mat{A}\herm = \mat{E}$ (see Theorem \ref{th:unilosslessisUnitary}). On the contrary, any unitary matrix $\mat{B}$ is unilossless (see Lemma \ref{th:AnyUnitaryIsUnilossless}). For any non-singular diagonal matrix $\mat{F}$, $\mat{A} = \mat{F}^{-1} \mat{B} \mat{F}$ is unilossless (see Theorem \ref{th:conjugateDiagonalPMInvariance}). Let $\mat{E} = \mat{F}^{-1}\mat{F}^{-\conj}$ and finally
			\begin{equation}
	 		\begin{aligned}
				\mat{A} \mat{E} \mat{A}\herm 
				= \mat{F}^{-1} \mat{B} \mat{F} \mat{E} \mat{F}^{\conj} \mat{B}\herm \mat{F}^{-\conj} 
				= \mat{E}.
			\end{aligned}
		 	\end{equation}
	\end{enumerate}
\end{proof}

Theorem \ref{th:sufficientNecessaryUnilossless} fully characterizes the set of unilossless matrices. The following section discusses briefly some aspects of non-unilossless matrices and their connection to lossless FDNs.
 
\subsection{Lossless, but not Unilossless}
This section discusses some aspects of lossless FDNs which are not based on unilossless feedback matrices. For the FDN to be lossless, following Theorem \ref{th:reciprocalCohn}, $p_{\mat{A},\vec{m}}(z)$ needs to be self-inversive and therefore $\mat{A}$ is characterized by Theorem \ref{th:diagonalSimilarInverseGeneral}. In the following, a parameterization of the unitary and conjugate-involutory matrices are given. 

\subsubsection{Characterizing the Unitary and Conjugate-Involutory Matrices}
Here we would like to give a characterization for all matrices $\mat{A}$ with $\mat{A}\mat{A}\herm = \mat{I}$ or $\mat{A}\mat{A}^\conj = \mat{I}$. Both classes can be identified with the exponential map. It is well known that for any matrix with $\mat{A}\mat{A}\herm = \mat{I}$, there is a skew-symmetric matrix $\mat{M}$ such that $\mat{A} = \exp(\mat{M})$ \cite{golub2012matrix}. A similar map can be shown for any matrix with $\mat{A}\mat{A}^\conj = \mat{I}$. Namely, there is a pure imaginary matrix $\mat{M}$ such that $\mat{A} = \exp(\mat{M})$. A proof is given in Appendix \ref{sec:ParameterizationConjugateInvolutory}. 

It can be seen that for any eigenvalue $\lambda_{\mat{A}}$ of $\mat{A}$ there is an eigenvalue $\lambda_{\mat{M}}$ of $\mat{M}$ with $\lambda_{\mat{A}} = \exp (\lambda_{\mat{M}})$. Therefore, $\mat{A}$ has only unimodular eigenvalues if and only if $\mat{M}$ has only pure imaginary eigenvalues. Interestingly, skew-symmetric matrices have only pure imaginary eigenvalues, whereas pure imaginary matrices clearly do not need to have only pure imaginary eigenvalues. The question when a pure imaginary matrix has only pure imaginary eigenvalues is a rather difficult question and to the best of our knowledge no simple criterion exists.

\subsubsection{On Lossless $2\times2$ FDNs}
We give some illustrations on matrices $\mat{A}$ which can constitute a lossless FDN for certain values of delays $\vec{m}$. Given the delays $\vec{m} = [m_1, m_2]$ and a $2\times2$ matrix $\mat{A}$, following Theorem \ref{th:charPoly}, the generalized characteristic polynomial is:
\begin{eqnarray}
	p_{\mat{A},\vec{m}}(z) = z^{m_1+m_2} - a_{22} z^{m_1} - a_{11} z^{m_2} + \det \mat{A}. 		
\end{eqnarray} 
Since $p_{\mat{A},\vec{m}}(z)$ is necessarily self-inversive, we have $\det \mat{A} = \epsilon$ with $|\epsilon| = 1$ and $a_{22}^\conj = \epsilon^\conj a_{11}$. The off-diagonal elements $a_{12}$ and $a_{21}$ have no direct influence on the generalized characteristic polynomial except the determinant.

If $\mat{A}$ is unitary, then $|a_{11}| \leq 1$ and $p_{\mat{A},\vec{m}}(z)$ is lossless for any $\vec{m}$. Nonetheless there are certain choices of delays $\vec{m}$ such that $p_{\mat{A},\vec{m}}(z)$ is lossless even with $|a_{11}|>1$. For example, for $m_1 = 1$, $m_2 = 2$:
\begin{equation}
\begin{aligned}
	p_{\mat{A},\vec{m}}(z) &= z^3 - a_{11} z^2 - \epsilon a_{11}^\conj z + \epsilon \\
	p'_{\mat{A},\vec{m}}(z) &= 3 z^2 - 2 a_{11} z - \epsilon a_{11}^\conj. 
\end{aligned}	
\end{equation}
The roots of $p'_{\mat{A},\vec{m}}(z)$ are
\begin{align}
	r_{1/2} = \frac{a_{11} \pm \sqrt{a_{11}^2 + 3 \epsilon a_{11}^\conj} }{3}.
\end{align}
According to Theorem \ref{th:reciprocalCohn},  $p_{\mat{A},\vec{m}}(z)$ is lossless if and only if all roots of $p'_{\mat{A},\vec{m}}(z)$ are on or within the unit circle, i.e., $|r_{1/2}| \leq 1$. For $\epsilon = -1$ and $a_{11} = 3$, the roots $r_{1/2} = 1$ and the corresponding FDN is lossless. This coincides with the example \eqref{eq:startCounterExamplePositiveCase} given in the beginning. Figure 
	\ref{fig:StabilityPolytopesFor2x2_det-1} depicts the stability region of $a_{11}$, i.e., the set of values $a_{11}$ for which $p_{\mat{A},\vec{m}}(z)$ is lossless under fixed $\det \mat{A}$ and $\vec{m}$. This plot shows the outer border of the stability region for $\vec{m} = [1, k]$ and $\det \mat{A} = -1$. It can be observed that for higher $k$, the stability regions approximate the unit circle, which is guaranteed by the unilossless condition. Also, this analysis illustrates that the practice of employing non-unilossless matrices $\mat{A}$ for the FDN design results in rather subtle and difficult determination of the lossless property. This can become especially difficult for long delays as the determination of the system poles suffers from numerical intractability.

\begin{figure} 
	\centering 
	\includegraphics[width=0.5\textwidth]{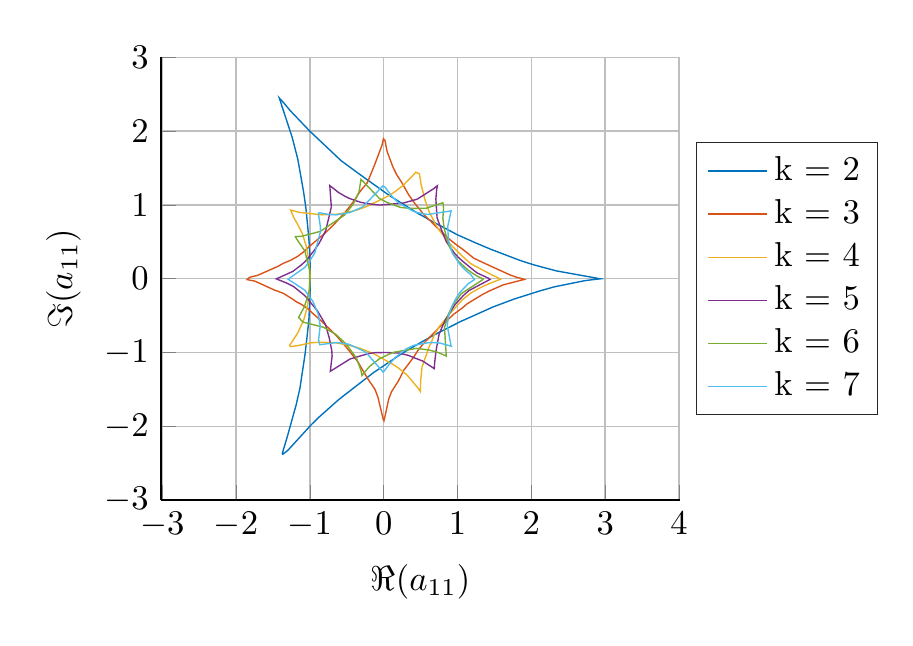} 
	\caption{Stability region of $a_{11}$ with $\det \mat{A} = -1$ and $\vec{m} = [1, k]$ for which $p_{\mat{A},\vec{m}}(z)$ is lossless. The depicted lines indicate the outer border of the valid set of values.} 
	\label{fig:StabilityPolytopesFor2x2_det-1}
\end{figure} 

\section{Unilossless Reverb Topologies}
\label{sec:LosslessTopologies}

Various topologies for delay based reverberation algorithms have been proposed over the last 50 years. As the delay lines can be reordered and the feedback paths can be summarized in an adjacency matrix, any of these topologies can be represented by an FDN structure. Whereas the computational efficiency of the original structure may be better than a general matrix computation, the system analysis benefits from a standardized representation.  In the following, we discuss some of the classic examples and their FDN versions.
	
\subsection{ Schroeder's Parallel Comb and Serial Allpass  }
In \cite{schroeder1962natural}, Schroeder presented an early topology for a delay based network. It follows a two stage design: four parallel comb filters to create the reverberation tail and then two serial allpass filters to increase the reflection density.

We define the main building block: the feedback-feedforward comb filter
\begin{equation}
C_{m, g_b, g_f}(z) = \frac{z^{-m} - g_f}{1 - g_b\, z^{-m}}.
\end{equation}
where $m$ is the delay in samples, $g_b$ and $g_f$ are the feedback and feedforward gains, respectively.
Comb allpass filters are feedback-feedforward comb filters with identical feedback and feedforward gains, i.e., $C_{m,g}(z) = C_{m, g, g}(z)$. The filter $C_{m,g}(z)$ has a unit magnitude response and therefore called allpass. The Schroeder structure is then given by
\begin{align}
	H_S(z) = \left( \sum_{i=1}^4 C_{m_i, g_{i}, 0}(z) \right) \prod_{j=5}^{6} C_{m_j,g_j}(z).
\end{align}	
The corresponding feedback matrix is 
\begin{align}
	\mat{A}_S = 
	\begin{bmatrix}
	g_1 & 0 & 0 & 0 & 0 & 0 \\
	0 & g_2 & 0 & 0 & 0 & 0 \\
	0 & 0 & g_3 & 0 & 0 & 0 \\
	0 & 0 & 0 & g_4 & 0 & 0 \\
	1 & 1 & 1 & 1 & g_5 & 0 \\
	-g_5 & -g_5 & -g_5 & -g_5 & 1-g_5^2 & g_6 	
	\end{bmatrix} \\
	\vec{c}_S = [ g_5 g_6, g_5 g_6, g_5 g_6, g_5 g_6, g_6(g_5^2 - 1), 1 - g_6^2]
\end{align}
with $\vec{m} = [m_1, \dots, m_6]$. The triangular feedback matrix $\mat{A}_S$ has the block matrix form
\begin{align}
	\mat{A}_S = 
	\begin{bmatrix}
	g_1 &  		&   &   &   	& \vec{0}\tran \\
	  	& g_2 	&  &   &  	& \vec{0}\tran \\
	  	&   	& g_3 &   &   & \vec{0}\tran \\
	  	&   &   & g_4 &   & \vec{0}\tran \\
	  	&   &   &   & g_5 & \vec{0}\tran \\
	\vec{g}'_1 & \vec{g}'_2 & \vec{g}'_3 	& \vec{g}'_4 & \vec{g}'_5 & g_6 
	\end{bmatrix} 
\end{align}
such that according to \eqref{eq:permutationCojugate}, $\mat{A}_S$ is reducible. According to  Condition 1 of Theorem~\ref{th:sufficientNecessaryUnilossless} is a reducible matrix unilossless if and only if the irreducible components $g_i$ are unilossless, i.e., $|g_i| = 1$ for all $i$. However for real $g_i$, the output from the $5^{th}$ and $6^{th}$ delays is zero, which is consistent to the fact that a lossless allpass filter cannot have a decaying tail.

\subsection{Dahl and Jot's Absorbent Allpass FDN}

\begin{figure}[!tb]
	\centering 
	\begin{subfigure}[]{0.48\textwidth}
		\includegraphics[width=\textwidth]{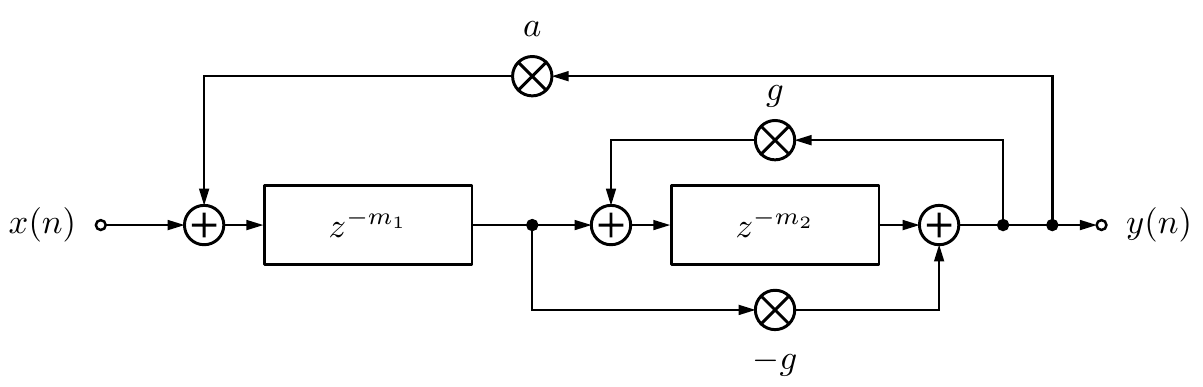}	
		\subcaption{Allpass delay Line with feedback}	
		\label{fig:allpassFeedbackLoop}
	\end{subfigure}
	\begin{subfigure}[]{0.48\textwidth}
		\includegraphics[width=\textwidth]{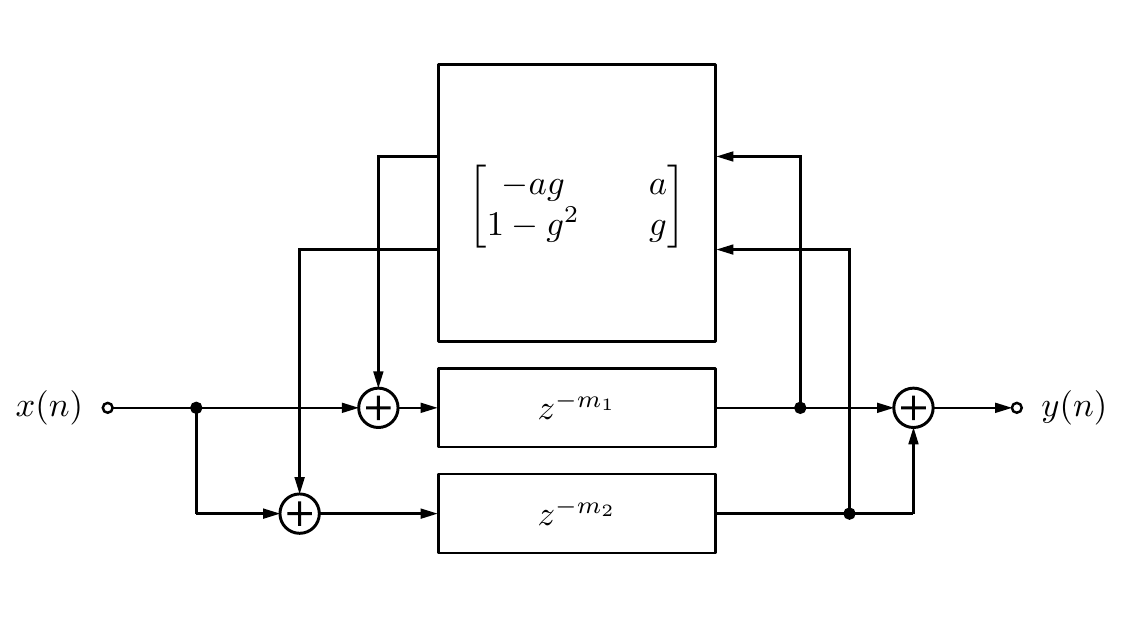}	
		\subcaption{Equivalent standard FDN}	
		\label{fig:allpassFeedbackFDN}
	\end{subfigure}
	\caption{An allpass filter in series with a delay filter in a feedback loop with feedback gain $a$ and the equivalent representation as a standard FDN.}
	\label{fig:allpassFeedback}
\end{figure}

The absorbent allpass FDN suggested in \cite{Dahl2000} extends every delay of the FDN by a feedback-feedforward allpass filter. Similar structures have been discussed also earlier \cite{Dattorro1997, VaananenEfficient97}. The open feedback loop is described by
\begin{align}
	\mat{D}_{\vec{m}}(z^{-1}) \mat{C}_{\vec{m'},\vec{g}}(z^{-1}) - \mat{A},
\end{align}
where $\mat{C}_{\vec{m'},\vec{g}}(z) = \diag{[ C_{m'_1,g_1}(z), \dots, C_{m'_N,g_N}(z) ]}$ is a diagonal matrix of allpass filters. In \cite{schlecht2014modulation}, the authors showed that any allpass FDN can be represented as a standard FDN with twice the delay lines and the feedback matrix
\begin{align}
	\mat{A}_{\textrm{AP}} =& \begin{bmatrix}-\mat{A}\,\mat{G}  & \mat{A}   \\ \mat{I} - \mat{G}^2 & \mat{G}\end{bmatrix},
\end{align}	
where $\mat{G} = \diag{[g_1, \dots, g_N]}$ and $\vec{m}_{AP} = [\vec{m}, \vec{m}']$. Figure \ref{fig:allpassFeedback} depicts a single delay line allpass FDN and an equivalent two delay line standard FDN. It has been shown in \cite{Dahl2000} that the allpass FDN is lossless if all $|g_i| \leq 1$. This can be confirmed by demonstrating that $\mat{A}_{\textrm{AP}}$ is diagonally similar to a unitary matrix:
\begin{equation}
	\begin{aligned}
	\mat{A}_{\textrm{AP}} =& \begin{bmatrix}-\mat{A}\,\mat{G}  & \mat{A}   \\ \mat{I} - \mat{G}^2 & \mat{G}\end{bmatrix} \\
	=& \begin{bmatrix} \mat{I}  & \mat{0}   \\ \mat{0} & \mat{G}' \end{bmatrix} 
	\begin{bmatrix}-\mat{A}\,\mat{G}  & \mat{A} \mat{G}'  \\ \mat{G}' & \mat{G}\end{bmatrix} 
	\begin{bmatrix} \mat{I}  & \mat{0}   \\ \mat{0} & \mat{G}'^{-1} \end{bmatrix},
	\label{eq:allpassMatrixN_diagonallysimilar}		
	\end{aligned}
\end{equation}
where $\mat{G}' = \sqrt{ \mat{I} - \mat{G}^2}$ and
\begin{align}
	\begin{bmatrix}-\mat{A}\,\mat{G}  & \mat{A} \mat{G}'  \\ \mat{G}' & \mat{G}\end{bmatrix} 
	= \begin{bmatrix} \mat{A}  & \mat{0} \\ \mat{0} & \mat{I} \end{bmatrix}  \begin{bmatrix}-\mat{G}  & \mat{G}' \\ \mat{G}' & \mat{G} \end{bmatrix} 
\end{align}
is unitary. Consequently, $\mat{A}_{\textrm{AP}}$ is diagonally similar to a unitary matrix and irreducible such that according to the Condition 2 of Theorem~\ref{th:sufficientNecessaryUnilossless} $\mat{A}_{\textrm{AP}}$ is unilossless. Whereas the implementation complexity might be beneficial for the allpass FDN, this analysis suggests that in terms of the sonic design of reverberation, unilossless FDNs offers the full potential of delay based systems. 

\subsection{ De Sena's Scattering Delay Network }
The scattering delay network (SDN) proposed in \cite{DeSena:2015bb, deSena2011scattering} extends the core structure of the FDN and draws the parameterization of the filter coefficients from a room geometry based analysis. Two different feedback matrix types are employed:
\begin{align}
	\mat{A}_{S1} = \frac{2}{\langle \vec{1}, \vec{y} \rangle} \vec{1} \vec{y}\tran - \mat{I}
	\textrm{ and }
	\mat{A}_{S2} = \frac{2}{\norm{\vec{y}}_2^2} \vec{y} \vec{y}\tran - \mat{I}
\end{align}
where $\vec{1} = [1, \dots, 1]\tran$, $\vec{y} = [y_1, \dots, y_N]\tran$ and $\langle \cdot, \cdot \rangle$ denotes the scalar product. Matrix $\mat{A}_{S2}$ is a Householder reflection, which is unitary and therefore a unilossless matrix (see Lemma \ref{th:AnyUnitaryIsUnilossless}). Matrix $\mat{A}_{S1}$ satisfies
\begin{align}
	\mat{A}_{S1}\herm \mat{Y} \mat{A}_{S1} = \mat{Y},
\end{align}
where $\mat{Y} = \diag{\vec{y}}$.  Consequently, $\mat{A}_{S1}$ is diagonally similar to a unitary matrix and according to the Condition 2 of Theorem.~\ref{th:sufficientNecessaryUnilossless}, $\mat{A}_{S1}$ is a unilossless feedback matrix. The matrix $\mat{A}_{S1}$ is therefore an example of an explicitly non-unitary, but unilossless matrix in application.

\section{Conclusion}
The present work proposed a new characterization of lossless FDNs based on unilossless feedback matrices, which are independent from the choice of the delays. Further, a sufficient and necessary condition for unilossless feedback matrices was given, which sets the complete framework for any FDN design based on lossless prototypes. Additionally, some aspects of lossless FDNs with non-unilossless feedback matrices were discussed. The proposed framework was applied to three existing artificial reverberation topologies suggesting a unified analysis approach of delay based physical modeling topologies based on FDNs.

\appendices
\section{Flawed Proof in \cite{Rocchesso96} }
\label{sec:FailedProofRocchesso}
By assuming that each delay line is longer than two samples, \eqref{eq:jotdifference} can be restated in the state-space description as follows \cite{Rocchesso96}:
\begin{equation}
\begin{aligned}
y(n) &=& \vec{\stateSpace{c}}^T \, \vec{\stateSpace{s}}(n) + {\stateSpace{d}} \, x(n) \\
\vec{\stateSpace{s}}(n+1) &=& \A \, \vec{\stateSpace{s}}(n) + \vec{\stateSpace{b}} \, x(n)
\label{eq:jotss}
\end{aligned}
\end{equation}
where
\begin{eqnarray}
{\stateSpace{d}} &=& d \\
\vec{\stateSpace{b}}^T &=& [0, ..., 0, \bold{b}^T] \\
\vec{\stateSpace{c}}^T &=& [0, ... , 0, \bold{c}^T, \underbrace{0, ..., 0}_{N}].
\end{eqnarray}

The state-transition matrix is
\begin{eqnarray} \label{eq:systembigmatrix}
\A &=& \begin{bmatrix} 
		\bold{U}_1 & 0 & 0 & \cdots & 0 & 0 & \bold{R}_1 \\ 
		0 & \bold{U}_2 & 0 & \cdots & 0 & 0 & \bold{R}_2 \\
		\vdots & \vdots & \vdots & \ddots & \vdots & \vdots & \vdots \\
		0 & 0 & 0 & \cdots & \bold{U}_N & 0 & \bold{R}_N \\
		\bold{P}_1 & \bold{P}_2 & \bold{P}_3 & \cdots & \bold{P}_N & 0 & 0 \\
		0 & 0 & 0 & \cdots & 0 & \bold{A} & 0 \\
		\end{bmatrix}
\end{eqnarray}
where
\begin{eqnarray}
\bold{U}_j	&=&	\begin{bmatrix} 
		0 & 1 & 0 & \cdots & 0 \\ 
		0 & 0 & 1 & \cdots & 0 \\
		\vdots & \vdots & \vdots & \ddots & \vdots \\
		0 & 0 & 0 & \cdots & 1 \\
		0 & 0 & 0 & \cdots & 0 \\
		\end{bmatrix}
\end{eqnarray}
\begin{eqnarray}
\bold{R}_j	&=&	\begin{bmatrix} 
		0 \cdots 0 & \cdots &  0 & \cdots & 0 \\ 
		0 \cdots 0 & \cdots &  0 & \cdots & 0 \\ 
		\vdots & \ddots & \vdots & \ddots & \vdots \\
		\underbrace{0 \cdots 0}_{j-1} & 1 & 0 & \cdots & 0 \\
		\end{bmatrix} 
\end{eqnarray}
and
\begin{eqnarray}
\bold{P}_j	&=&	\begin{matrix}
  \coolleftbrace{j-1}{ 0 \\ \vdots \\0 } \\ \vphantom{0} \\  \vphantom{0} \\ \vphantom{\vdots} \\ \vphantom{0} \\
  \end{matrix}%
  \begin{bmatrix} 0 & 0 & \cdots & 0 \\ 
		\vdots & \vdots & \ddots & \vdots \\
		0 & 0 & \cdots & 0 \\
		1 & 0 & \cdots & 0 \\
		0 & 0 & \cdots & 0 \\
		\vdots & \vdots & \ddots & \vdots \\
		0 & 0 & \cdots & 0 \\ 
  \end{bmatrix}%
\end{eqnarray}
Let us define
\begin{eqnarray}
\bold{U}_j &\in& \mathbb{C}^{(m_j -2)\times (m_j -2)} \\
\bold{R}_j &\in& \mathbb{C}^{(m_j -2)\times N}\\
\bold{P}_j &\in& \mathbb{C}^{N\times(m_j -2)}. 
\end{eqnarray}
Given the so-called generalized admittance
\begin{align}
	\widetilde{\mat{\Gamma}} = 
	\begin{bmatrix}
	\mat{I}_{\N-N} & \mat{0} \\ \mat{0} & \mat{\Gamma} 	
	\end{bmatrix},
\end{align}
where $\mat{\Gamma} \in \mathbb{C}^{N\times N} $ is a Hermitian, positive-definite matrix, we have
\begin{align}
	\A\tran \widetilde{\mat{\Gamma}} \A = \diag{[ \mat{I}_{m_1-2}, \dots, \mat{I}_{m_N-2}, \mat{B}]}.
\end{align}
The proof fails in (29) of \cite{Rocchesso96} as
\begin{align}
	\mat{B} = 
	\begin{bmatrix}
		\mat{A}^T \mat{\Gamma} \mat{A} & \mat{0} \\ \mat{0} & \mat{I}
	\end{bmatrix}
	\textrm{ instead of }
	\mat{B} = 
	\begin{bmatrix}
		\mat{I} & \mat{0} \\ \mat{0} & \mat{A}^T \mat{\Gamma} \mat{A}
	\end{bmatrix}.
\end{align}

\section{Proof of Theorem \ref{th:diagonalSimilarInverseGeneral} }
\label{sec:ProofOfHartfielLoewyInverse}

We first recall two main results on which the proof of Theorem~\ref{th:diagonalSimilarInverseGeneral} is based upon.
 
\begin{theorem}[Hartfiel, \cite{Hartfiel1984}]
	Let $N=2$ or $N=3$ and $\mat{A}$ is irreducible. Then if $\mat{A}$ and $\mat{B}$ have equal principal minors then there is a non-singular diagonal matrix $\mat{E}$ with either $\mat{E}^{-1} \mat{A} \mat{E} = \mat{B}$ or $\mat{E}^{-1} \mat{A} \mat{E} = \mat{B}\tran$. 
\label{th:Hartfiel23}
\end{theorem}

Consequently for $N < 4$, all matrices $\mat{A}^{-1}$ with property \eqref{eq:reciprocalPrincipalMinors} are diagonally similar to $\mat{B} = \mat{A}^\conj$ or $\mat{B}\tran = \mat{A}\herm$. For higher $N$, an additional condition has to be posed on $\mat{A}$ for the general case.

\begin{theorem}[Loewy, \cite{loewy1986principal}]
	Let $N \geq 4$ and $\mat{A}$ is irreducible, and for every partition of $\{1,2,...,N\}$ into subsets $\alpha, \beta$ with $|\alpha|\geq 2, |\beta| \geq 2$ either $\rank \mat{A}(\alpha|\beta) \geq 2$ or $\rank \mat{A}(\beta|\alpha) \geq 2$.
	
	 Then if $\mat{A}$ and $\mat{B}$ have equal principal minors then there is a non-singular diagonal matrix $\mat{E}$ with either $\mat{E}^{-1} \mat{A} \mat{E} = \mat{B}$ or $\mat{E}^{-1} \mat{A} \mat{E} = \mat{B}\tran$. 
	 \label{th:loewyDS}
\end{theorem}

The following example shows that the additional rank property in Theorem~\ref{th:loewyDS} is not necessary. The matrix 
\begin{equation}
	\mat{A}=\frac{1}{5}
\begin{bmatrix}
 -1 & 4 & -2 & -2 \\
 -4 & 1 & 2 & 2 \\
 2 & 2 & -1 & 4 \\
 -2 & -2 & -4 & 1 \\
\end{bmatrix}
\end{equation}
is orthogonal and for this $\mat{A}$ and ${\mat{A}^{-}}\tran$ are diagonally similar despite the fact that the rank property is violated by $\rank \mat{A}(1 2 | 3 4) = 1$. The following shows that the rank property can be dropped in Theorem~\ref{th:loewyDS} for the special case of $\mat{B} = \mat{A}^{-*}$. Please refer to the original proof in \cite{Hartfiel1984} for additional technical information on the given proof. Similarly to the proof of Theorem~\ref{th:loewyDS}, the following proof is by induction on the matrix dimension $N$, where as the base case is $N = 4$. The following lemma, which is a combination of Lemma~1 and 4 in \cite{Hartfiel1984}, provides the inductive step:
\begin{lemma}
	Let $N \geq 4$ and suppose $\mat{A}, \mat{B} \in \mathbb{C}^{N\times N}$ are irreducible. Let
	\begin{equation*}
		\begin{aligned}
			S_1 &= \cbracket{i : 1 \leq i \leq N : \mat{B}(i^c) \textrm{ is diagonally similar to } \mat{A}(i^c)} \\
			S_2 &= \cbracket{i : 1 \leq i \leq N : \mat{B}(i^c)\tran \textrm{ is diagonally similar to }  \mat{A}(i^c)}. \\
		\end{aligned}
	\end{equation*}
	Then,
	\begin{enumerate}
		\item if $|S_1| \geq 3$, then $\mat{B}$ is diagonally similar to $\mat{A}$;
		\item if $|S_2| \geq 3$, then $\mat{B}$ is diagonally similar to $\mat{A}$;
		\item if $|S_1 \bigcup S_2 | \geq 5$, then either $\mat{B}$ or $\mat{B}\tran$ is diagonally similar to $\mat{A}$.
	\end{enumerate}
\end{lemma} 
Let $N \geq 5$. The inductive assumption provides that for irreducible $\mat{A}, \mat{B} \in \mathbb{C}^{N-1 \times N-1}$, either  $\mat{B}$ or $\mat{B}\tran$ is diagonally similar to $\mat{A}$. Hence, for sets $S_1$ and $S_2$ as defined in Lemma~6, $| S_1 \bigcup S_2 | \geq 5$. Therefore, the induction is valid by Lemma~6. It remains to proof the base case $N = 4$.

Let $N = 4$. The proof of Theorem~4 in \cite{Hartfiel1984} applies Theorem~\ref{th:Hartfiel23} to the $3 \times 3$ principal submatrices and arrives at the following representation based solely on $\mat{A}$ being irreducible and $\mat{A}$ and $\mat{B}$ having equal corresponding principal minors and \emph{without} exploiting the additional rank property:

\begin{align}
\begin{aligned}
		\mat{A} &= \begin{bmatrix}
		a_{11} & a_{12} & a_{13} & a_{14} \\
		a_{21} & a_{22} & a_{23} & a_{24} \\
		a_{13} f_3 & \frac{a_{23} f_3}{f_2} & a_{33} & a_{34} \\
		a_{14} e_4 & \frac{a_{24} e_4}{f_2} & a_{43} & a_{44} \\
	\end{bmatrix}
\end{aligned}
\end{align}
and
\begin{align}
\begin{aligned}
	 {\mat{B}} &= \begin{bmatrix}
		a_{11} & \frac{a_{21}}{f_2} & a_{13} & a_{14} \\
		a_{12}f_2 & a_{22} & a_{23} & a_{24} \\
		a_{13}f_3 & \frac{a_{23}f_3}{f_2} & a_{33} & a_{34} \\
		a_{14}e_4 & \frac{a_{24}e_4}{f_2} & a_{43} & a_{44} \\
	\end{bmatrix}, \\
	 {{\mat{B}}}\tran &= \begin{bmatrix}
		a_{11} & {a_{12}}{f_2} & a_{13} f_3 & a_{14} e_4 \\
		\frac{a_{21}}{ f_2 } & a_{22} & \frac{a_{23} f_3}{f_2} & \frac{a_{24} e_4}{f_2} \\
		a_{13} & a_{23} & a_{33} & a_{43} \\
		a_{14} & a_{24} & a_{34} & a_{44} \\
	\end{bmatrix},
	\label{eq:startCondition}
\end{aligned}
\end{align}
where $f_2, f_3, e_3, e_4 \neq 0$ are coefficients from diagonal similarity transformations. Because $\mat{B}$ is by construction diagonally similar to $\mat{A}^{-*}$, Jacoby's identity \eqref{eq:jacoby} yields
\begin{equation}
\begin{aligned}
	\det \mat{A}( 1 3 | 2 4 ) &= \det \mat{B} ( 2 4 | 1 3 ) \\
	\det \begin{bmatrix}
		a_{12} & a_{14} \\
		\frac{a_{23}f_3}{f_2} & a_{34}
	\end{bmatrix} &= 
	\det \begin{bmatrix}
		a_{12} f_2 & a_{23} \\
		a_{14}e_4 & a_{43}
	\end{bmatrix} \\
	a_{12} a_{34} - a_{14} \frac{a_{23}f_3}{f_2} &=
	a_{12} f_2 a_{43}  - a_{23} a_{14}e_4 \\
	a_{12} ( a_{34} - f_2 a_{43} )  &=
	a_{14} a_{23} \paren{ \frac{f_3}{f_2} - e_4 } 
\end{aligned}
\end{equation}
and 
\begin{equation}
\begin{aligned}
	\det \mat{A}( 2 4 | 1 3 ) &= \det \mat{B} ( 1 3 | 2 4 ) \\
	\det \begin{bmatrix}
		a_{21} & a_{23} \\
		a_{14}e_4 & a_{43}
	\end{bmatrix} &= 
	\det \begin{bmatrix}
		\frac{a_{21}}{ f_2} & a_{14} \\
		\frac{a_{23}f_3}{ f_2} & a_{34}
	\end{bmatrix} \\
	a_{21} a_{43} - a_{23} a_{14}e_4 &=
	\frac{a_{21}}{ f_2} a_{34}  - a_{14} \frac{a_{23}f_3}{ f_2} \\
	a_{21} ( a_{43} - \frac{ a_{34} }{ f_2 } )  &=
	- a_{14} a_{23} \paren{ \frac{f_3}{f_2} - e_4 } .
\end{aligned}
\end{equation}
The combination of both identities is then
\begin{equation}
	\begin{aligned}
		a_{12} ( a_{34} - f_2 a_{43} )  &= - a_{21} ( a_{43} - \frac{ a_{34} }{ f_2 } )  \\
		a_{12} f_2 ( a_{34} - f_2 a_{43} )  &=  a_{21}  ( a_{34} - f_2 a_{43} )  \\
		(a_{12} f_2 - a_{21})( a_{34} - f_2 a_{43} )  &=  0 .
	\end{aligned}
\end{equation}
If the first term $(a_{12} f_2 - a_{21}) = 0$, then $\mat{A} = \mat{B}$. If the second term $( a_{34} - f_2 a_{43} ) = 0$, then $\mat{B}\tran = \diag{ 1, f_2, f_3, e_4}^{-1} \mat{A} \diag{ 1, f_2, f_3, e_4}$. Thus, the matrices $\mat{A}$ and $\mat{B}$ or $\mat{B}\tran$ are diagonally similar.

\section{Parameterization of Conjugate-Involutory Matrices} 
\label{sec:ParameterizationConjugateInvolutory}
\begin{lemma}
	For any $\mat{A}$ with $\mat{A}\mat{A}^\conj=\mat{I}$, there is a real matrix $\mat{M}$ with $\mat{A} = \exp(\imath \mat{M})$, where $\imag = \sqrt{-1}$.
	\label{th:ParameterizationConjugateInvolutory}
\end{lemma}
\begin{proof}
	Let $\mat{M}=\mat{U}+\imath \mat{V}$, where $\mat{U}$, $\mat{V}$ are real. $\mat{A}\mat{A}^\conj=\mat{I}$ is equivalent to $\mat{U}^2+\mat{V}^2=\mat{I}$ and $\mat{U}\mat{V}=\mat{V}\mat{U}$. Consequently, there is a real matrix $\mat{M}$ with $\mat{U}=\cos(\mat{M})$ and $\mat{V} = \sin(\mat{M})$ such that $\mat{A}=\exp(\imath \mat{M})$.	
\end{proof}

\section*{Acknowledgment}
The authors would like to thank a number of researchers for their helpful comments and discussions: G. H. E. Duchamp for the alternative proof of Lemma \ref{th:AnyUnitaryIsUnilossless}, V. Protsak for pointing out the diagonal similarity invariance, R. Bryant and L. Blanc for the proof of Lemma \ref{th:ParameterizationConjugateInvolutory}, R. Loewy and D. J. Hartfiel for helpful comments on their work and D. Rocchesso for his fundamental contribution to FDN design and supportive discussion of the present idea.


\end{document}